\documentclass[reqno]{amsart}
\usepackage{color}
\usepackage{amsmath,amsthm}
\usepackage{amsfonts}

\usepackage[all]{xy}
\usepackage{graphicx}

\definecolor{myurlcolor}{rgb}{0.6,0,0}
\definecolor{mycitecolor}{rgb}{0,0,0.8}
\definecolor{myrefcolor}{rgb}{0,0,0.8}
\usepackage[pagebackref]{hyperref}
\hypersetup{colorlinks,
linkcolor=myrefcolor,
citecolor=mycitecolor,
urlcolor=myurlcolor}

\newcommand{\FinStoch}{\mathtt{FinStoch}}
\newcommand{\FinStat}{\mathtt{FinStat}}
\newcommand{\FinProb}{\mathtt{FinProb}}
\newcommand{\FP}{\mathtt{FP}}
\newcommand{\Top}{\mathtt{Top}}
\newcommand{\Set}{\mathtt{Set}}
\newcommand{\Cat}{\mathtt{Cat}}
\newcommand{\E}{\mathtt{E}}

\newcommand{\tends}{\to}       
\newcommand{\define}[1]{{\bf \boldmath{#1}}}
\newcommand{\maps}{\colon}
\newcommand{\RE}{\mathrm{RE}}
\renewcommand{\P}{\mathrm{P}}
\renewcommand{\O}{\mathrm{O}}

\newcommand{\R}{\mathbb{R}}

\newcommand{\supp}{\mathrm{supp}}

\newcommand{\beq}{\begin{equation}}
\newcommand{\eeq}{\end{equation}}

\theoremstyle{plain}

\newtheorem{thm}{Theorem}
\newtheorem*{unthm}{Unproved ``Theorem''}
\newtheorem{lem}[thm]{Lemma}
\newtheorem{prop}[thm]{Proposition}

\newtheorem{defn}[thm]{Definition}

\theoremstyle{definition}

\theoremstyle{remark}

\numberwithin{equation}{section}

\begin{document}



\title[Relative entropy]{A Bayesian characterization \\ of relative entropy}

\author{John C.~Baez}
\address{Department of Mathematics\\ 
University of California\\ 
Riverside CA 92521\\
USA \\
and Centre for Quantum Technologies\\ 
National University of Singapore\\ 
Singapore 117543}
\email{baez@math.ucr.edu}
\author{Tobias Fritz}
\address{Perimeter Institute for Theoretical Physics \\
31 Caroline St. N, Waterloo, Ontario N2L 2Y5, Canada}
\email{tfritz@perimeterinstitute.ca}



\begin{abstract}
We give a new characterization of relative entropy, also known as the Kullback--Leibler divergence. We use a number of interesting categories related to probability theory.   In particular, we consider a category $\FinStat$ where an object is a finite set equipped with a probability distribution, while a morphism is a measure-preserving function $f \maps X \to Y$ together with a stochastic right inverse $s \maps Y \to X$.   The function $f$ can be thought of as a measurement process, while $s$ provides a hypothesis about the state of the measured system given the result of a measurement.  Given this data we can define the entropy of the probability distribution on $X$ relative to the `prior' given by pushing the probability distribution on $Y$ forwards along $s$.  We say that $s$ is `optimal' if these distributions agree.  We show that any convex linear, lower semicontinuous functor from $\FinStat$ to the additive monoid $[0,\infty]$ which vanishes when $s$ is optimal must be a scalar multiple of this relative entropy.   Our proof is independent of all earlier characterizations, but inspired by the work of Petz.  
\end{abstract}

\maketitle

\setcounter{tocdepth}{1}
\tableofcontents

\section{Introduction}
\label{introduction}

This paper gives a new characterization of the concept of relative entropy, also known as `relative information', `information gain' or `Kullback-Leibler divergence'.   Whenever we have two probability distributions $p$ and $q$ on the same finite set $X$, we define the information of $q$ relative to $p$ as:
\[
S(q,p) = \sum_{x\in X} q_x \ln\left( \frac{q_x}{p_x} \right) 
\]
Here we set $q_x \ln(q_x/p_x)$ equal to $\infty$ when $p_x = 0$, unless $q_x$ is also zero, in which case we set it equal to 0.  Relative entropy thus takes values in $[0,\infty]$.  

Intuitively speaking, $S(q,p)$ is the expected amount of information gained when we discover the probability distribution is really $q$, when we had thought it was $p$.   We should think of $p$ as a `prior'.  When we take $p$ to be the uniform distribution on $X$, relative entropy reduces to the ordinary Shannon entropy, up to a sign and an additive constant. The advantage of relative entropy is that it makes the role of the prior explicit.
   
Since Bayesian probability theory emphasizes the role of the prior, relative entropy naturally lends itself to a Bayesian interpretation~\cite{IB}. Our goal here is to make this 
precise in a mathematical characterization of relative entropy.  We do this using a 
category $\FinStat$ where:
\begin{itemize}
\item an object $(X,q)$ consists of a finite set $X$ and a probability distribution $x\mapsto q_x$ on that set;
\item a morphism $(f,s) \maps (X,q) \to (Y,r)$ consists of a measure-preserving 
function $f$ from $X$ to $Y$, together with a probability distribution $x\mapsto s_{xy}$ on $X$ for each element $y \in Y$ with the property that $s_{xy} = 0$ unless $f(x) = y$.
\end{itemize}
We can think of an object of $\FinStat$ as a system with some finite set of 
\define{states} together with a probability distribution on its states.   A morphism $(f,s) 
\maps (X,q) \to (Y,r)$  then consists of two parts.  First, there is a deterministic 
`measurement process' $f \maps X \to Y$ mapping states of some system being measured to states of a `measurement apparatus'.  The condition that $f$ be measure-preserving says that the probability that the apparatus winds up in some state $y \in Y$ 
is the sum of the probabilities of states of $X$ leading to that outcome:
\[     r_y = \sum_{x : \; f(x) = y} q_x  .\]
Second, there is a `hypothesis' $s$: an assumption about the probability 
$s_{xy}$ that the system being measured is in the state $x$ given any 
measurement outcome $y \in Y$.  We assume that this probability vanishes 
unless $f(x) = y$, as we would expect from a hypothesis made by someone who knew the behavior of the measurement apparatus.

Suppose we have any morphism $(f,s) \maps (X,q) \to (Y,r)$  in $\FinStat$.   From 
this we obtain two probability distributions on the states of the system being measured.   First, we have the probability distribution $p \maps X \to \R$ given by
\begin{equation}
\label{eq:prior}
 p_x = s_{x\, f(x)} r_{f(x)}.
\end{equation}
This is our `prior', given our hypothesis and the probability distribution of measurement outcomes.  Second, we have the `true' probability distribution $q \maps X \to \R$.  It follows that any morphism in $ \FinStat$ has a relative entropy $S(q,p)$ associated to it.   This is the expected amount of information we gain when we update our prior $p$ to $q$.

In fact, this way of assigning relative entropies to morphisms defines a functor
\[  \RE \maps \FinStat \to [0,\infty] \] 
where we use $[0,\infty]$ to denote the category with one object, the nonnegative real numbers together with $\infty$ as morphisms, and addition as composition.   More precisely, if $(f,s) \maps (X,q) \to (Y,r)$ is any morphism in $\FinStat$, we define
\[   \RE(f,s) = S(q,p)  \]
where the prior $p$ is defined as in Equation \eqref{eq:prior}.  The fact that $\RE$ is a functor is nontrivial and rather interesting.   It says that given any composable pair of measurement processes:
\[    (X,q) \stackrel{(f,s)}{\longrightarrow} (Y,r) \stackrel{(g,t)}{\longrightarrow} (Z,u) \]
the relative entropy of their composite is the sum of the relative entropies
of the two parts:
\[     \RE((g,t) \circ (f,s)) = \RE(g,t) + \RE(f,s) .\]
We prove that $\RE$ is a functor in Section \ref{sec:charent}.  However, we go much further: we characterize relative entropy by saying that up to a constant multiple, $\RE$ is the unique functor from $\FinStat$ to $[0,\infty]$ obeying three reasonable conditions.

The first condition is that $\RE$ vanishes on morphisms  $(f,s) \maps (X,q) \to (Y,r)$   where the hypothesis $s$ is `optimal'.  By this, we mean that Equation \eqref{eq:prior} gives a prior $p$ equal to the `true' probability distribution $q$ on the states of the system being measured.  

The second condition is that $\RE$ is lower semicontinuous. The set $P(X)$ of probability distributions on a finite set $X$ naturally has the topology of an $(n-1)$-simplex when $X$ has $n$ elements.   The set $[0,\infty]$ can be given the topology induced by the usual order on this set, and it is then homeomorphic to a closed interval.   However, with these topologies, the relative entropy does not define a continuous function
\[       \begin{array}{rcl}         S \maps P(X) \times P(X) &\to& [0,\infty]  \\
                                                          (q,p) &\mapsto & S(q,p) .
\end{array}
\]
The problem is that 
\[
S(q,p) = \sum_{x\in X} q_x \ln\left( \frac{q_x}{p_x} \right) 
\]
and $q_x \ln(q_x/p_x)$ equals $\infty$ when 
$p_x = 0$ and $q_x > 0$, but $0$ when $p_x = q_x = 0$.  So, it turns out that $S$ is only lower semicontinuous, meaning that it can suddenly jump down, but not up.
More precisely,  if $p^i , q^i \in P(X)$ are sequences with $p^i \to p$, $q^i \to q$,
then
\[       S(q,p) \le \liminf_{i \to \infty} S(q^i, p^i)  .\]
In Section \ref{sec:charent} we give the set of morphisms in $\FinStat$ a topology, and show that with this topology, $\RE$ maps morphisms to morphisms in a lower semicontinuous way.  

The third condition is that $\RE$ is convex linear.   In Section \ref{sec:charent} we describe how to take convex linear combinations of morphisms in $\FinStat$.   
The functor $\RE$ is convex linear in the sense that it maps any convex linear combination of morphisms in $\FinStat$ to the corresponding convex linear combination of numbers in $[0,\infty]$.  Intuitively, this means that if we flip a probability-$\lambda$ coin to decide whether to perform one measurement process or another, the expected information gained is $\lambda$ times the expected information gain of the first process plus $(1 - \lambda)$ times the expected information gain of the second.

Our main result is Theorem \ref{thm1}: any lower semicontinuous, convex linear 
functor
\[  F \maps \FinStat \to [0,\infty] \]
that vanishes on morphisms with an optimal hypothesis must equal some constant 
times the relative entropy.  In other words, there exists some constant $c \in [0,\infty]$ such that 
\[      F(f,s) = c\, \RE(f,s)  \]
for any morphism $(f,s) \maps (X,p) \to (Y,q)$ in $\FinStat$.

This theorem, and its proof, was inspired by results of Petz \cite{Petz}, who sought
to characterize relative entropy both in the `classical' case discussed here and in the more general `quantum' setting.  Our original intent was merely to express his results in a more category-theoretic framework.  Unfortunately his work contained a flaw, which we had to repair.  As a result, our proof is now self-contained.  For details, see the remarks after Theorem \ref{thm:petz}.

Our characterization of relative entropy implicitly relies on topological categories and on the operad whose operations are convex linear combinations.  However, since these structures are not strictly necessary for stating or proving our result, and they may be unfamiliar to some readers, we discuss them only in Appendix \ref{app:semicont} and Appendix \ref{app:convalgs}.

\section{The categories in question} 
\label{categories}

\subsection{FinStoch}
To describe the categories used in this paper, we need to start with a word on the category of finite sets and stochastic maps.  A stochastic map $ f \maps X \leadsto Y$ is different from an ordinary function, because instead of assigning a unique element of $ Y$ to each element of $ X,$ it assigns a \textit{probability distribution} on $ Y$ to each element of $ X.$   Thus $ f(x)$ is not a specific element of $Y,$ but instead has a probability of taking on different values. This is why we use a wiggly arrow to denote a stochastic map.

More formally:

\begin{defn}  Given finite sets $X$ and $Y$, a \define{stochastic map} $ f \maps X \leadsto Y$ assigns a real number $ f_{yx}$ to each pair $ x \in X, y \in Y$ in such a way that fixing any element $x$, the numbers $ f_{yx}$ form a probability distribution on $ Y.$  We call $ f_{yx}$ \define{the probability of $ y$ given $x$. }  \end{defn}

In more detail, we require that the numbers $f_{yx}$ obey:

\begin{itemize}
\item
$ f_{yx} \ge 0$ for all $ x \in X,$ $ y \in Y$, 
\vskip 1em
\item
$ \displaystyle{ \sum_{y \in Y} f_{yx} = 1}$ for all $ x \in X.$ 
\end{itemize}
Note that we can think of $ f \maps X \leadsto Y$ as a $Y \times X$-shaped matrix of numbers.  A matrix obeying the two properties above is called \define{stochastic}.  This viewpoint is nice because it reduces the problem of composing stochastic maps to matrix multiplication.  It is easy to check that multiplying two stochastic matrices gives a stochastic matrix.  So, we define the composite of stochastic maps $f \maps X \leadsto Y$ and $g \maps Y \leadsto Z$ by
\[        (g \circ f)_{zx} = \sum_{y \in Y} g_{zy} f_{yx}  .\]
Since matrix multiplication is associative and identity matrices are stochastic, this construction gives a category:

\begin{defn}  Let \define{$\FinStoch$} be the category of finite sets and stochastic maps between them.  
\end{defn}

We are restricting attention to finite sets merely to keep the discussion simple and avoid issues of convergence. It would be interesting to generalize all our work to more general probability spaces.

\subsection{FinProb}
\label{FinProb}

Choose any 1-element set and call it $\mathbf{1}$.  A function $f \maps \mathbf{1} \to X$ is just a point of $X$.  But a stochastic map $q \maps \mathbf{1} \leadsto X$ is something more interesting: it is a probability distribution on $X$.

We use the term \define{finite probability measure space} to mean a finite set with a probability distribution on it.  As we have just seen, there is a very quick way to describe such a thing within $\FinStoch$:
\[  
\xymatrix{ \mathbf{1} \ar@{~>}[dd]_q \\ \\
X   
} 
\]
This gives a quick way to think about a measure-preserving function between finite probability measure spaces!  It is simply a commutative triangle like this:
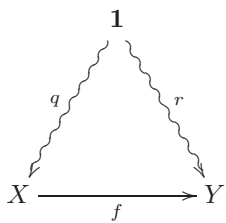
\[ 
\xymatrix{  & \mathbf{1} \ar@{~>}[ddl]_q \ar@{~>}[ddr]^r \\ \\
X \ar[rr]_f & & Y  
} 
\]
Note that the horizontal arrow $ f \maps X \to Y$ is not wiggly. The straight arrow means it is an honest function, not a stochastic map.  But a function can be seen as a special case of a stochastic map.  So it makes sense to compose a straight arrow with a wiggly arrow---and the result is, in general, a wiggly arrow.  If we then demand that the above triangle commute, this says that the function $ f \maps X \to Y$ is measure-preserving.  

We now work through the details. First: how can we see a function as a special case of a stochastic map?  A function $ f \maps X \to Y$ gives a matrix of numbers
\[ f_{yx} = \delta_{y\, f(x)} \]
where $\delta$ is the Kronecker delta.  This matrix is stochastic, and it defines a stochastic map sending each point $ x \in X$ to the probability distribution supported at $f(x)$. 

Given this, we can see what the commutativity of the above triangle means.  If we use $		q_x$ to stand for the probability that $q \maps 1 \leadsto X$ assigns to each element
$x\in X$, and similarly for $r_y$, then the triangle commutes if and only if
\[ \displaystyle{ r_y = \sum_{x \in X} \delta_{y \, f(x)} q_x} \] 
or in other words:
\[  \displaystyle{ r_y = \sum_{x : \; f(x) = y} q_x } \]
In this situation we say $p$ is $q$ \define{pushed forward along} $f$, and that $f$ is a \define{measure-preserving function}.

So, we have used $ \FinStoch$ to describe another important category:

\begin{defn}  Let \define{$\FinProb$} be the category of finite probability measure spaces and measure-preserving functions between them. \end{defn}

Another variation may be useful at times:
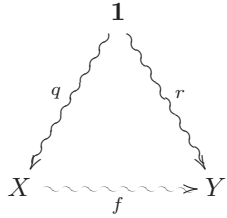
\[ 
\xymatrix{  & \mathbf{1} \ar@{~>}[ddl]_q \ar@{~>}[ddr]^r \\ \\
X \ar@{~>}[rr]_f & & Y  
} 
\]
A commuting triangle like this is a \define{measure-preserving stochastic map}.  In other words, $q$ gives a probability measure on $ X,$ $r$ gives a probability measure on $Y$, and $ f \maps X \leadsto Y$ is a stochastic map that is measure-preserving in the following sense:
\[ \displaystyle{ r_y = \sum_{x \in X} f_{yx} q_x } . \]

\subsection{FinStat} 
\label{FinStat}

The category we need for our characterization of relative entropy is a bit more subtle.  In this category, an object is a finite probability measure space:
\[  
\xymatrix{ \mathbf{1} \ar@{~>}[dd]_q \\ \\
X   
} 
\]
but a morphism looks like this:
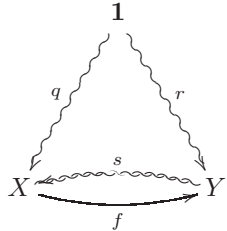
\[ \xymatrix{  & \mathbf{1} \ar@{~>}[ddl]_q \ar@{~>}[ddr]^r \\ \\
X \ar@/_/[rr]_f & & Y  \ar@{~>}@/_/[ll]_s
}\]
\[  \begin{array}{ccl}      f \circ q &=& r \\  
        f \circ s &=& 1_Y  \end{array}  \]
The diagram need not commute, but the two equations shown must hold.  The first equation says that $ f \maps X \to Y$ is a measure-preserving function.  In other words, this triangle, which we have seen before, commutes:
\[ 
\xymatrix{  & \mathbf{1} \ar@{~>}[ddl]_q \ar@{~>}[ddr]^r \\ \\
X \ar[rr]_f & & Y  
} 
\]
The second equation says that $ f \circ s$ is the identity, or in other words, $s$ is a `section' for $f$.   This requires a bit of discussion.

We can think of $X$ as the set of `states' of some system, while $Y$ is a set of possible states of some other system: a `measuring apparatus'.  The function $f$ is a `measurement process'.  One `measures' the system using $f$, and if the system is in any state $x \in X$ the measuring apparatus goes into the state $f(x)$.  The probability distribution $q$ gives the probability that the system is in any given state, while $r$ gives the probability that the measuring apparatus ends up in any given state after a measurement is made.

Under this interpretation, we think of the stochastic map $s$ as a `hypothesis' about the system's state given the state of the measuring apparatus.  If one measures the system and the apparatus goes into the state $ y \in Y,$ this hypothesis asserts that the system is in the state $x$ with probability $s_{xy}.$  

The equation $f \circ s = 1_Y$ says that if the measuring apparatus ends up in some state $y \in Y$, our hypothesis assigns a nonzero probability only to states of the measured system for which a measurement actually leads to this state $y$:

\begin{lem} 
\label{lem:finstat}
If $f \maps X \to Y$ is a function between finite sets and $s \maps Y \leadsto X$ is a stochastic map, then $f \circ s = 1_Y$ if and only for all $y \in Y$,
$s_{xy} = 0$ unless $f(x) = y$.
\end{lem}

\begin{proof}  The condition $f \circ s = 1_Y$ says that for any fixed $y,y' \in Y$,
\[ \sum_{x : \; f(x) = y'} s_{xy} = \sum_{x \in X} \delta_{y'\, f(x)} s_{xy} = \delta_{y' y} .\]
It follows that the sum at left vanishes if $y' \ne y$.  If $s$ is stochastic, the terms in this sum are nonnegative.  So, $s_{xy}$ must be zero if $f(x) = y'$ and $y' \ne y$.

Conversely, suppose we have a stochastic map $s \maps Y \leadsto X$ such that $s_{xy} = 0$ unless $f(x) = y.$    Then for any $y \in Y$ we have
\[ 1 = \sum_{x \in X} s_{xy} =  \sum_{x : \; f(x) = y} s_{xy} = 
 \sum_{x \in X} \delta_{y\, f(x)} s_{xy} \]
while for $y' \ne y$ we have
\[ 0 = \sum_{x : \; f(x) = y'} s_{xy} = 
 \sum_{x \in X} \delta_{y'\, f(x)} s_{xy} ,\]
so for all $y, y' \in Y$
\[    \sum_{x \in X} \delta_{y'\, f(x)} s_{xy} =  \delta_{y' y} ,\]
which says that $f \circ s = 1_Y$.  
\end{proof}

It is also worth noting that $f \circ s = 1_Y$ implies that $f$ is onto: if $y \in Y$ were not in the image of $f$, we could not have
\[  \sum_{x \in X} s_{xy} = 1 \]
as required, since $s_{xy} = 0$ unless $f(x) = y.$   So, the equation $f \circ s = 1_Y$
also rules out the possibility that our measuring apparatus has `extraneous' states that never arise when we make a measurement.

This is how we compose morphisms of the above sort:
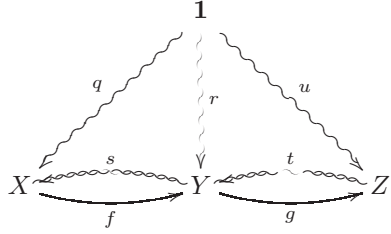
\[ \xymatrix{  && \mathbf{1} \ar@{~>}[ddll]_q  \ar@{~>}[dd]^r \ar@{~>}[ddrr]^u \\ \\
X \ar@/_/[rr]_f &&  Y  \ar@{~>}@/_/[ll]_s   \ar@/_/[rr]_g  &&   
Z  \ar@{~>}@/_/[ll]_t
}\]
\[   \begin{array}{cclcccl}       f \circ q &=& r & \;  & g \circ r &=& u  \\
       f \circ s &=& 1_Y   &\; & g \circ t &=& 1_Z \end{array} \]
We get a measure-preserving function $ g \circ f \maps X \to Z$ and a stochastic map going back, $ s \circ t \maps Z \to X.$  It is easy to check that these obey the required equations:
\[ g \circ f \circ q = u \]
\[ g \circ f \circ s \circ t = 1_Z \]
So, this way of composing morphisms gives a category, which we call
$\FinStat$, to allude to its role in statistical reasoning:

\begin{defn}  Let \define{$\FinStat$} be the category where an object is a finite probability measure space:
\[  
\xymatrix{ \mathbf{1} \ar@{~>}[dd]_q \\ \\
X   
} 
\]
a morphism is a diagram 
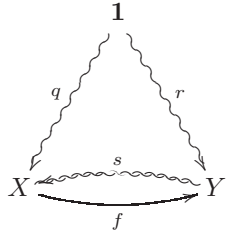
\[ \xymatrix{  & \mathbf{1} \ar@{~>}[ddl]_q \ar@{~>}[ddr]^r \\ \\
X \ar@/_/[rr]_f & & Y  \ar@{~>}@/_/[ll]_s
}\]
obeying these equations:
\[  \begin{array}{ccl}      f \circ q &=& r \\  
        f \circ s &=& 1_Y  \end{array}  \]
and composition is defined as above.
\end{defn}

\subsection{FP}
\label{FP}

We have described how to think of a morphism in $\FinStat$ as consisting of a `measurement process' $ f$ and a `hypothesis' $s$, obeying two equations: 
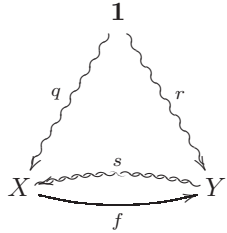
\[ \xymatrix{  & \mathbf{1} \ar@{~>}[ddl]_q \ar@{~>}[ddr]^r \\ \\
X \ar@/_/[rr]_f & & Y  \ar@{~>}@/_/[ll]_s
}\]
\[  \begin{array}{ccl}      f \circ q &=& r \\  
        f \circ s &=& 1_Y  \end{array}  \]
We say the hypothesis is \define{optimal} if also
\[ s \circ r = q .\] 
Conceptually, this says that if we take the probability distribution $r$ on our observations and use it to infer a probability distribution for the system's state using our hypothesis 
$s,$ we get the correct answer: $q$.   Mathematically, it says that this diagram commutes:
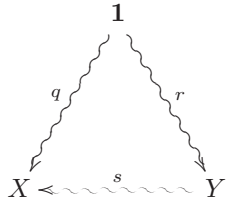
\[ \xymatrix{  & \mathbf{1} \ar@{~>}[ddl]_q \ar@{~>}[ddr]^r \\ \\
X  & & Y  \ar@{~>}[ll]_s
}\]
In other words, $s$ is a measure-preserving stochastic map.

It is easy to check that this optimality property is preserved by composition of morphisms. Hence there is a subcategory of $\FinStat$ with all the same objects, but only morphisms where the hypothesis is optimal:

\begin{defn}  Let \define{$\FP$} be the subcategory of $\FinStat$ where an object is a finite probability measure space
\[  
\xymatrix{   \mathbf{1} \ar@{~>}[dd]_q \\ \\
X   
} 
\]
and a morphism is a diagram 
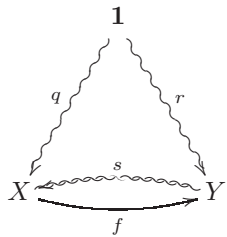
\[ \xymatrix{  & \mathbf{1} \ar@{~>}[ddl]_q \ar@{~>}[ddr]^r \\ \\
X \ar@/_/[rr]_f & & Y  \ar@{~>}@/_/[ll]_s
}\]
obeying these equations:
\[  \begin{array}{ccl}      f \circ q &=& r \\  
        f \circ s &=& 1_Y \\
        s \circ r &=& q 
 \end{array}  \]
\end{defn}

The category $\FP$ was introduced by Leinster \cite{Leinster}.  He
gave it this name for two reasons.  First, it is a close relative of $\FinProb$, where 
a morphism looks like this:
\[ 
\xymatrix{  & \mathbf{1} \ar@{~>}[ddl]_q \ar@{~>}[ddr]^r \\ \\
X \ar[rr]_f & & Y  
} 
\]
We now explain the similarities and differences between $\FP$ and $\FinProb$ by studying the properties of the forgetful functor $\FP\to\FinProb$, which sends every morphism $(f,s)$ to its underlying measure-preserving function $f$.

For a morphism in $\FP$, the conditions on $s$ are so strong that they completely determine it, unless there are states of the measurement apparatus that happen with probability zero: that is, unless there are $y \in Y$ with $r_y = 0$.  To see this, note that 
\[ s \circ r = q \] 
says that
\[ \sum_{y \in Y} s_{xy} r_y = q_x \]
for any choice of $x \in X.$   But we have already seen in Lemma \ref{lem:finstat}
that $s_{xy} = 0$ unless $f(x) = y$, so the sum has just one term, and the equation says
\[ s_{xy} r_{y} = q_x \]
where $y = f(x)$. We can solve this for $s_{xy}$ unless $r_y = 0$.  Furthermore, we have already seen that every $y \in Y$ is of the form $f(x)$ for some $x \in X$.

Thus, for a morphism $(f,s) \maps (X,q) \to (Y,r)$ in $\FP$, we can solve for $s$ in terms of the other data unless there exists $y \in Y$ with $r_y = 0$.   Except for this special case, a morphism in $\FP$ is just a morphism in $\FinProb$.  But in this special case, a morphism in $\FP$ has a little extra information: an arbitrary probability distribution on the inverse image of each point $y$ with $r_y = 0$.  The point is that in $\FinStat$, and thus $\FP$, a `hypothesis' must provide a probability for each state of the system given a state of the measurement apparatus, even for states of the measurement apparatus that occur with probability zero.

A more mathematical way to describe the situation is that our functor $\FP\to\FinProb$ is `generically' full and faithful: the function 
\[
\begin{array}{ccc}
\FP((X,q),(Y,r)) & \longrightarrow& \FinProb((X,q),(Y,r)) \\
    (f,s)            & \mapsto & f 
\end{array}
\]
is a bijection if the support of $r$ is the whole set $Y$, which is the generic situation.

The second reason Leinster called this category $\FP$ is that it is  
freely formed from an operad called $\P$. This is a topological operad whose $n$-ary 
operations are probability distributions on the set $\{1, \dots, n\}$.  These
operations describe convex linear combinations, so algebras of this operad include convex subsets of $\mathbb{R}^n$, more general convex spaces \cite{Fritz}, 
and even more.   As Leinster explains \cite{Leinster}, the category $\FP$
(or more precisely, an equivalent one) 
is the 
free $\P$-algebra among categories containing an internal 
$\P$-algebra.  We will not need this fact here, but it
is worth mentioning that Leinster used this fact to characterize entropy as a 
functor from $\FP$ to $[0,\infty)$.  He and the authors then rephrased
this in simpler language \cite{BFL}, obtaining a characterization of entropy as a functor from $\FinProb$ to $[0,\infty)$.  The characterization of relative entropy
in the current paper is a closely related result.  However, the proof is completely different.

\section{Characterizing entropy}
\label{sec:charent}

\subsection{The theorem}

We begin by stating our main result.  Then we clarify some of the terms involved and 
begin the proof.

\begin{thm}
\label{thm1}
Relative entropy determines a functor
\beq
\label{entropyfunctor}
\begin{array}{ccl}
\RE \maps \FinStat &\to& [0,\infty]  \\  \\
\bigg( \xymatrix{ (X,q) \ar@/_/[rr]_f && (Y,r) \ar@{~>}@/_/[ll]_s } \bigg)  &\mapsto& S(q,s\circ r) \\
\end{array} \eeq
that is lower semicontinuous, convex linear, and vanishes on morphisms in the subcategory $\FP$.  

Conversely, these properties characterize the functor $\RE$ up to a scalar multiple.  In other words, if $F$ is another functor with these properties, then for some $0 \le c \le \infty$ we have $F(f,s) = c\, \RE(f,s)$ for all morphisms $(f,s)$ in $\FinStat$.  (Here we define $\infty \cdot a = a \cdot \infty = \infty$ for $0 < a \le \infty$, but $\infty \cdot 0 = 0 \cdot \infty = 0$.)
\end{thm}

In the rest of this section we begin by describing $[0,\infty]$ as a category and checking that $\RE$ is a functor.  Then we describe what it means for the functor $\RE$ to be lower semicontinuous and convex linear, and check these properties.  We postpone the hard part of the proof, in which we characterize $\RE$ up to a scalar multiple by these properties, to Section \ref{sec:proof}.

In what follows, it will be useful to have an explicit formula for $S(q,s \circ r)$.  By definition,
\[           S(q, s \circ r) = \sum_{x \in X} q_x \ln\left(\frac{q_x}{(s \circ r)_x} \right) \]
We have
\[           (s \circ r)_x = \sum_{y \in Y} s_{xy} r_y ,\]
but by Lemma \ref{lem:finstat}, $s_{xy} = 0$ unless $f(x) = y$, so the sum has just one term:
\[           (s \circ r)_x = s_{x\, f(x)} r_{f(x)} \]
and we obtain
\begin{equation}
\label{eq:relative_entropy}
           S(q, s \circ r) = \sum_{x \in X} q_x \ln\left(\frac{q_x}{s_{x \, f(x)} r_{f(x)}} 
\right) .
\end{equation}

\subsection{Functoriality}

We make $[0,\infty]$ into a monoid using addition, where we define 
addition in the usual way for numbers in $[0,\infty)$ and set
\[             \infty + a = a + \infty = \infty \]
for all $a \in [0,\infty]$.   There is thus a category with one object and elements of $[0,\infty]$ as endomorphisms of this object, with composition of morphisms given by addition.  With a slight abuse of language we also use $[0,\infty]$ to denote this category.

\begin{lem}
\label{lem:functor}
The map $\RE \maps \FinStat \to [0,\infty]$ described in Theorem \ref{thm1} is a functor. 
\end{lem}

\begin{proof} 
Let
\[
\xymatrix{ (X,q) \ar@/_/[rr]_f && \ar@{~>}@/_/[ll]_s (Y,r) \ar@/_/[rr]_g && (Z,u) \ar@{~>}@/_/[ll]_t }
\]
be a composable pair of morphisms in $\FinStat$. Then the functoriality of $\RE$ can
be shown by repeated use of Equation \eqref{eq:relative_entropy}:
\[ 
\begin{array}{ccl}
\RE\left(g\circ f,s\circ t\right) &=& S\left(q,s\circ t\circ u\right) \\ \\
&=& \displaystyle{ \sum_{x\in X} q_x\ln\left(\frac{q_x}{s_{x\, f(x)}t_{f(x)\,g(f(x))} u_{g(f(x))}} \right) } \\  \\
&\stackrel{(\ast)}{=}&  
\displaystyle{\sum_{x\in X} q_x\ln\left(\frac{q_x}{s_{x\, f(x)} r_{f(x)}}\right) + 
\sum_{x\in X} q_x \ln\left(\frac{r_{f(x)}}{t_{f(x)\,g(f(x))} u_{g(f(x))}}\right)}
\\  \\
&=& S(q, s \circ r) +  \displaystyle{ \sum_{y\in Y} r_y\ln\left(\frac{r_y}{t_{y\,g(y)} u_{g(y)}}\right) }
\\   \\
&=& S(q, s \circ r) +  S(r,t\circ u)  \\  \\
&=& \RE(f,s) + \RE(g,t) .
\end{array}
\]
Here the main step is $(\ast)$, where we have simply inserted 
\[ 0=  \sum_x q_x\ln\frac{1}{r_{f(x)}} + \sum_x q_x\ln r_{f(x)}.  \]
This is unproblematic as long as $r_{f(x)}>0$ for all $x$. When there are $x$ with $r_{f(x)}=0$, then we necessarily have $q_x=0$ as well, and both $q_x\ln\tfrac{1}{r_{f(x)}}$ and $q_x\ln r_{f(x)}$ actually vanish, so this case is also fine.  In the step after $(\ast)$, we use the fact that for each $y \in Y$, $r_y$ is the sum
of $q_x$ over all $x$ with $f(x) = y$.  \end{proof}

\subsection{Lower semicontinuity}
\label{sec:semicont}

Next we explain what it means for a functor
to be lower semicontinuous, and prove that $\RE$ has this property.   There is a way to think about semicontinuous functors in terms of topological categories, but this is not
really necessary for our work, so we postpone it to Appendix \ref{app:semicont}.  Here we take a more simple-minded approach.

If we fix two finite sets $X$ and $Y$, the set of all morphisms 
\[    (f,s) \maps (X,q) \to (Y,p) \]
in $\FinStat$ forms a topological space in a natural way.  To see this, let
\[   P(X) = \{ q \maps X \to [0,1] : \; \sum_{x \in X} q_x = 1 \} \]
be the set of probability distributions on a finite set $X$.  This is a subset of 
a finite-dimensional real vector space, so we give it the subspace
topology.  With this topology, $P(X)$ is homeomorphic to a simplex.  
The set of stochastic maps $s \maps Y \leadsto X$ is also a subspace of
a finite-dimensional real vector space, namely the space of matrices
$\mathbb{R}^{X \times Y}$, so we also give it the subspace topology.  We then give 
$P(X) \times P(Y) \times \mathbb{R}^{X \times Y}$ the product topology.  The set of morphisms $(f,s) \maps (X,q) \to (Y,p)$ in $\FinStat$ can be seen as a subspace of this, and we give it the subspace topology.  We then say:

\begin{defn} A functor $F \maps \FinStat \to [0,\infty]$ is \define{lower semicontinuous} 
if for any sequence of morphisms $(f, s^i) \maps (X,q^i) \to (Y,r^i)$ that converges
to a morphism $(f,s) \maps (X,q) \to (Y,r)$, we have
\[       F(f,s) \le \liminf_{i \to \infty} F(f,s^i)  .\] 
\end{defn}

\noindent
We could use nets instead of sequences here, but it would make no difference.  
We can then check another part of our main theorem:

\begin{lem} 
\label{lem:semicont}
The functor $\RE \maps \FinStat \to [0,\infty]$ described in Theorem \ref{thm1} is lower semicontinuous.
\end{lem}

\begin{proof}
Suppose that $(f, s^i) \maps (X,q^i) \to (Y,r^i)$ is a sequence of morphisms in $\FinStat$ that converges to $(f,s) \maps (X,q) \to (Y,r)$.   We need to show that
\[       S(q,s\circ r) \le \liminf_{i \to \infty} S(q^i,s^i\circ r^i)  .\] 
If there is no $x\in X$ with $s_{x\,f(x)} r_{f(x)}=0$ then this is clear, since all the elementary functions involved in the definition of relative entropy are continuous 
away from $0$. If all $x\in X$ with $s_{x\, f(x)}=0$ also satisfy $q_x=0$, then $S(q,s\circ r)$ is still finite since none of these $x$ contribute to the sum for $S$.  In this case $S(q^i, s^i \circ r^i)$ may remain arbitrarily large, even infinite as $i \to \infty$.  But the inequality
\[       S(q,s\circ r) \le \liminf_{i \to \infty} S(q^i,s^i\circ r^i)  \] 
remains true.  The same argument applies if there are $x\in X$ with $r_{f(x)}=0$, which implies $q_x=0$. Finally, if there are $x\in X$ with $s_{x\, f(x)}=0$ but $r_{f(x)} \geq q_x > 0$, then $S(q,s\circ r)=\infty$.   The above inequality is still valid in this
case.
\end{proof}

That lower semicontinuity of relative entropy is an important property was already known to Petz; see the closing remark in~\cite{Petz}.  

\subsection{Convex linearity}
\label{sec:convex_linearity}

Next we explain what it means to say that relative entropy gives 
a convex linear functor from $\FinProb$ to $[0,\infty]$, and we prove this is true.
In general, convex linear functors go between convex categories.  These are topological categories equipped with an action of the operad $\P$ discussed by Leinster \cite{Leinster}.  Since we do not need the general theory here, we postpone it
to Appendix~\ref{app:convalgs}. 

First, note that there is a way to take convex linear combinations of objects and
morphisms in $\FinProb$. Let $(X, p)$ and $(Y, q)$ be finite sets equipped
with probability measures, and let $\lambda \in [0, 1]$. Then there is a
probability measure
\[
\lambda p \oplus (1 - \lambda) q
\]
on the disjoint union $X + Y$, whose value at a point $x$ is
given by
\[
(\lambda p \oplus (1 - \lambda) q)_x
=
\begin{cases}
\lambda p_x        &\text{if } x \in X\\
(1 - \lambda) q_x &\text{if } x \in Y.
\end{cases}
\]
Given a pair of morphisms 
\[ f \maps (X,p) \to (X',p'), \qquad g \maps (Y,q) \to (Y',q') \]
in $\FinProb$, there is a unique morphism
\[
\lambda f \oplus (1 - \lambda) g \maps
(X + Y,\lambda p \oplus (1 - \lambda) q) \;\; \to \;\;
(X' + Y', \lambda p' \oplus (1 - \lambda) q')
\]
that restricts to $f$ on $X$ and to $g$ on $Y$.

A similar construction applies to $\FinStat$.  Given a pair of morphisms
\[
\xymatrix{ (X,p)\ar@/_/[rr]_f && (X',p') \ar@{~>}@/_/[ll]_s } \qquad \xymatrix{ (Y,q)\ar@/_/[rr]_g && (Y',q') \ar@{~>}@/_/[ll]_t } 
\]
in $\FinStat$, we define their convex linear combination to be
\[
\xymatrix{ (X + Y,\lambda p \oplus (1-\lambda)q) \ar@/_1.5pc/[rr]_{\lambda f\oplus (1-\lambda)g} && (X' + Y',\lambda p' \oplus (1-\lambda)q') \ar@{~>}@/_1.5pc/[ll]_{s\oplus t} }
\]
where $s\oplus t \maps X' + Y'\leadsto X + Y$ is the stochastic map which restricts to $s$ on $X'$ and $t$ on $Y'$. As a stochastic matrix, it is of block-diagonal form. It is right inverse to $\lambda f\oplus (1-\lambda)g$ by construction.

We may also define convex linear combinations of objects and morphisms
in the category $[0,\infty]$.   Since this category has only one object, there is only
one way to define convex linear combinations of objects.  Morphisms in this category are elements of the set $[0,\infty]$.  We have already made this set into a monoid using addition.  We can also introduce multiplication, defined in the usual way for numbers in $[0,\infty)$, and with 
\[            0  a = a 0 = 0  \]
for all $a \in [0,\infty]$.   This gives meaning to the convex linear combination
$\lambda a + (1 - \lambda) b$ of two morphisms $a,b$ in $[0,\infty]$.
For more details, see Appendices \ref{app:semicont} and \ref{app:convalgs}.
 
\begin{defn}
A functor $F \maps \FinStat \to [0,\infty]$ is \define{convex linear} if it preserves convex combinations of objects and morphisms.
\end{defn}

For objects this requirement is trivial, so all this really means is that for any pair of morphisms $(f,s)$ and $(g,t)$ in $\FinStat$ and any $\lambda \in [0,1]$, we have
\[
F\left(\lambda (f,s)\oplus (1-\lambda) (g,t)\right) \; = \; \lambda F(f,s) + (1-\lambda) F(g,t) .
\]

\begin{lem}
The functor $\RE \maps \FinStat \to [0,\infty]$ described in Theorem \ref{thm1} is
convex linear.
\end{lem}

\begin{proof} This follows from a direct computation:\smallskip
\begin{align*}
\RE((\lambda (f,s) \oplus &(1-\lambda)(g,t)) = S(\lambda p\oplus (1-\lambda)q,\lambda s\circ p'\oplus (1-\lambda) t\circ q') \\[4pt]
&= \sum_{x\in X} \lambda p_x \ln\left( \frac{\lambda p_x}{s_{x\,f(x)}\cdot\lambda p'_{f(x)}} \right) 
 + \sum_{y\in Y} (1-\lambda) q_y \ln\left( \frac{(1-\lambda) q_y}{t_{y\,g(y)}\cdot(1-\lambda) q'_{g(y)}} \right) \\[4pt]
& = \lambda \sum_{x\in X} p_x\ln\left( \frac{p_x}{s_{x\,f(x)} p'_{f(x)}} \right) + (1-\lambda) \sum_{y\in Y} \ln\left( \frac{q_y}{t_{y\,g(y)}q'_y} \right) \\[4pt]
& = \lambda S(p,s\circ p')+(1-\lambda) S(q,t\circ q') \\
& = \lambda \,\RE(f,s) + (1-\lambda) \, \RE(g,t)   \qedhere
\end{align*}
\end{proof}

\section{Proof of the theorem}
\label{sec:proof}

Now we prove the main part of Theorem \ref{thm1}.
 
\begin{lem}
\label{lem:converse}
Suppose that a functor
\[
F\maps \FinStat  \to [0,\infty]  
\]
is lower semicontinuous, convex linear, and vanishes on morphisms in the subcategory $\FP$.  Then for some $0 \le c \le \infty$ we have $F(f,s) = c\,\RE(f,s)$ for all morphisms $(f,s)$ in $\FinStat$.  
\end{lem}

\begin{proof} 
Let $F \maps \FinStat\to[0,\infty]$ be any functor satisfying these hypotheses.
By functoriality and the fact that $0$ is the only morphism in $[0,\infty]$ with an inverse, $F$ vanishes on isomorphisms.  Thus, given any commutative square in $\FinStat$ where the vertical morphisms are isomorphisms:
\[
\xymatrix{ (X,p) \ar@/_/[rr]_f \ar@/_/[dd]_{\rotatebox{90}{$\sim$}} 
&& (Y,q) \ar@/_/[dd]_{\rotatebox{90}{$\sim$}} \ar@{~>}@/_/[ll]_s \\\\
(X',p') \ar@{~>}@/_/[uu]_{\rotatebox{90}{$\sim$}} \ar@/_/[rr]_{f'} 
&& (Y',q') \ar@{~>}@/_/[ll]_{s'}  \ar@{~>}@/_/[uu]_{\rotatebox{90}{$\sim$}}
}
\]
functoriality implies that $F$ takes the same value on the top and bottom morphisms:
\[  F(f,s) = F(f',s'). \]
So, in what follows, we can replace an object by an isomorphic object without changing the value of $F$ on morphisms from or to this object.

Given any morphism in $\FinStat$, complete it to a diagram of this form:
\[
\xymatrix{ (X,p) \ar@/_/[rr]_f \ar@/_/[rdd]_{!_X} && (Y,q) \ar@/_/[ldd]_{!_Y} \ar@{~>}@/_/[ll]_s \\\\
& (\mathbf{1},1) \ar@{~>}@/_/[uul]_{s\circ q} \ar@{~>}@/_/[uur]_{q} }
\]
Here $\mathbf{1}$ denotes any one-element set equipped with the unique probability measure $1$, and $!_X \maps X\to\mathbf{1}$ is the unique function, which is automatically measure-preserving since $p$ is assumed to be normalized. Since this diagram commutes, and the morphism on the lower right lies in $\FP$, we obtain
\[
F \,\bigg( \xymatrix{ (X,p)\ar@/_/[rr]_f && (Y,q) \ar@{~>}@/_/[ll]_{s} } \bigg) = F\, \bigg( \xymatrix{ (X,p) \ar@/_/[rr]_{!_X} && (\mathbf{1},1) \ar@{~>}@/_/[ll]_{s\circ q} } \bigg) .
\]
In other words: the value of $F$ on a morphism depends only on the two distributions $p$ and $s\circ q$ living on the domain of the morphism. For this reason, it is enough to prove the claim only for those morphisms whose codomain is $(\mathbf{1},1)$.

We now consider the family of distributions
\[
q(\alpha) = \left(\alpha,1-\alpha\right),
\]
on a two-element set $\mathbf{2}=\{0,1\}$, and consider the function
\beq
\label{g}
g(\alpha) = F  \bigg( \xymatrix{ (\mathbf{2},q(1)) \ar@/_/[rr]_{!_\mathbf{2}} && (\mathbf{1},1) \ar@{~>}@/_/[ll]_{q(\alpha)} } \bigg)
\eeq
for $\alpha\in[0,1]$. Note that for all $\beta\in[0,1)$, this square in $\FinStat$ commutes:
\[
\xymatrix@!=1cm{ (\mathbf{3},(1,0,0)) \ar@/_.4cm/[dd]_*+\txt{$\scriptstyle{0\mapsto 0}$\\$\scriptstyle{1,2\mapsto 1}$} \ar@/_.4cm/[rr]_*+\txt{$\scriptstyle{0,1\mapsto 0}$\\$\scriptstyle{2\mapsto 1}$} && (\mathbf{2},(1,0)) \ar@/_.4cm/[dd]_{!_\mathbf{2}} \ar@{~>}@/_.4cm/[ll]_{q(\beta)\oplus 1} \\\\
 (\mathbf{2},(1,0)) \ar@{~>}@/_.4cm/[uu]_{1\oplus q\left(\frac{\alpha (1 - \beta)}{1 - \alpha\beta} \right)} \ar@/_.4cm/[rr]_{!_\mathbf{2}} && (\mathbf{1},1) \ar@{~>}@/_.4cm/[uu]_{q(\alpha)} \ar@{~>}@/_.4cm/[ll]_{q(\alpha\beta)} }
\]
where the left vertical morphism is in $\FP$, while the top horizontal
morphism is the convex linear combination
\[   1 \left(\xymatrix{ (\mathbf{2},q(1)) \ar@/_/[rr]_{!_\mathbf{2}} && (\mathbf{1},1) \ar@{~>}@/_/[ll]_{q(\beta)} }\right) \oplus 0 \left( 1_{(\mathbf{1},1)} \right) .\]
Applying the functoriality and convex linearity of $F$ to this square, we thus obtain the equation
\begin{equation}
\label{eq:cauchy}
g(\alpha\beta) = g(\alpha) + g(\beta) .
\end{equation}
We claim that all solutions of this equation are of the form $g(\alpha)=-c \ln \alpha$ for some $c\in [0,\infty]$.   First we show this for $\alpha \in (0,1]$.  

If $g(\alpha) < \infty$ for all $\alpha \in (0,1]$, this equation is Cauchy's functional equation in its multiplicative-to-additive form, and it is known \cite{Kuczma} that any solution with $g$ measurable is of the desired form for some $c < \infty$.   By our hypotheses on $F$, $g$ is lower semicontinuous, hence measurable. Thus, for some $c < \infty$ we have $g(\alpha)= -c \ln\alpha$ for all $\alpha \in (0,1].$

If $g(\alpha) = \infty$ for some $\alpha\in (0,1]$, then Equation \eqref{eq:cauchy} implies that $g(\beta)=\infty$ for all $\beta<\alpha$.  Since it also implies that $g(2\beta)=\tfrac{1}{2}g(\beta)$, we conclude that then $g(\beta)=\infty$ for all $\beta \in (0,1)$.  Thus, if we take $c = \infty$ we
again have $g(\alpha)= -c \ln\alpha$ for all $\alpha \in (0,1].$

Next consider $\alpha = 0$.  If $c>0$, then $g(0)=g(0)+g(\tfrac{1}{2})$ shows that we necessarily have $g(0)=\infty$. If $c=0$, then lower semicontinuity implies $g(0)=0$. In both cases, the equation $g(\alpha)=-c \ln\alpha$ also holds for $\alpha=0$. 

In what follows, choosing the value of $c$ that makes $g(\alpha) = -  c\ln \alpha$, we shall prove that the equation
\[
F\left( \xymatrix{ (X,p) \ar@/_/[rr]_{!_X} && (\mathbf{1},1) \ar@{~>}@/_/[ll]_r } \right) = c\, S(p,r)
\]
holds for any two probability distributions $p$ and $r$ on any finite set $X$.  
Using Equation \eqref{eq:relative_entropy}, it suffices to show that
\begin{equation}
\label{showF}
F\left( \xymatrix{ (X,p) \ar@/_/[rr]_{!_X} && (\mathbf{1},1) \ar@{~>}@/_/[ll]_r } \right) = c \sum_{x \in X} p_x \ln \left(\frac{p_x}{r_x} \right) .
\end{equation}
We prove this for more and more general cases in the following
series of lemmas.  We start with the generic case, where $c < \infty$ and the probability distribution $r$ has full support.   In Lemma \ref{lem:4.3} we treat all cases with $0 < c < \infty$.   In Lemma \ref{lem:4.4} we treat the case 
$c = 0$, and in Lemma \ref{lem:claim6} we treat the case $c = \infty$,  which seems much harder than the rest.
\end{proof}

\begin{lem}
\label{lem:4.1}
Equation~\eqref{showF} holds if $c < \infty$ and the support of $r$ is all of $X$.
\end{lem}

\begin{proof}
Choose $\alpha\in(0,1)$ such that $\alpha < r_x$ for all $x \in X$. The decisive step is to consider the commutative square
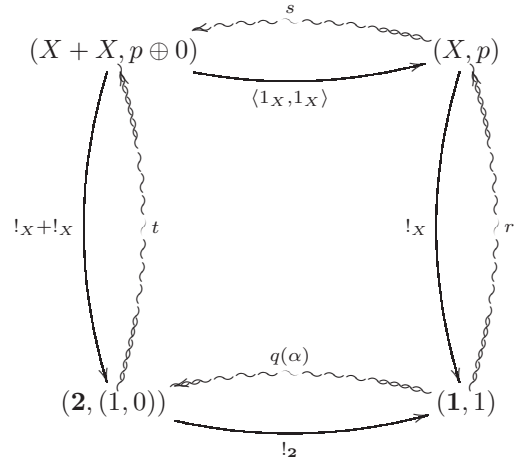
\[
\xymatrix@!=1.5cm{ (X + X, p\oplus 0) \ar@/_.4cm/[rr]_{\langle 1_X,1_X\rangle} \ar@/_.4cm/[dd]_{!_X + !_X} && (X,p) \ar@/_.4cm/[dd]_{!_X} \ar@{~>}@/_.4cm/[ll]_s \\\\
 (\mathbf{2},(1,0)) \ar@/_.4cm/[rr]_{!_{\mathbf{2}}} \ar@{~>}@/_.4cm/[uu]_t && (\mathbf{1},1) \ar@{~>}@/_.4cm/[uu]_r \ar@{~>}@/_.4cm/[ll]_{q(\alpha)} }
\]
where the stochastic matrices $s$ and $t$ are given by
\[
s = 
\left(\begin{array}{ccc}
\alpha \tfrac{p_1}{r_1} &&  \makebox(0,0){\text{\huge0}} \\ & \ddots &  \\  \makebox(0,0){\text{\huge0}}  && \alpha \tfrac{p_n}{r_n} \\ \\ 1-\alpha\tfrac{p_1}{r_1} &&  \makebox(0,0){\text{\huge0}}  \\ &\ddots& \\ 
\makebox(0,0){\text{\huge0}}  && 1-\alpha\tfrac{p_n}{r_n} \end{array}\right)  ,
\qquad t = 
\left(\begin{array}{cc}
p_1 & \frac{r_1 -  \alpha p_1}{1-\alpha}\\\vdots & \vdots\\ 
p_n & \frac{r_n -\alpha p_n}{1-\alpha}\end{array}\right).
\]
The second column of $t$ is only relevant for commutativity. The left vertical morphism is in $\FP$, while we already know that the lower horizontal morphism evaluates to $g(\alpha)=-c \ln\alpha$ under the functor $F$. Hence the diagonal of the square gets assigned the value $-c \ln\alpha$ under $F$.
On the other hand, the upper horizontal morphism is actually a convex linear combination of morphisms
\[
\xymatrix{ (\mathbf{2},(1,0)) \ar@/_.4cm/[rr]_{!_\mathbf{2}} && (\mathbf{1},1) \ar@{~>}@/_.4cm/[ll]_{q\left(\alpha\tfrac{p_x}{r_x}\right)} }  ,
\]
one for each $x \in X$, with the probabilities $p_x$ as coefficients.  Thus, composing this with the right vertical morphism we get a morphism
to which $F$ assigns the value
\[
-c\sum_{x\in X} p_x \ln\left(\alpha\frac{p_x}{r_x}\right) + F\left( \xymatrix{ (X,p) \ar@/_/[rr]_{!_X} && (\mathbf{1},1) \ar@{~>}@/_/[ll]_r } \right) .
\]
Thus, we obtain
\[
-c\sum_{x\in X} p_x \ln\left(\alpha\frac{p_x}{r_x}\right) + F\left( \xymatrix{ (X,p) \ar@/_/[rr]_{!_X} && (\mathbf{1},1) \ar@{~>}@/_/[ll]_r } \right) = - c \ln \alpha
\]
and because $c < \infty$, we can simplify this to
\[   
F\left( \xymatrix{ (X,p) \ar@/_/[rr]_{!_X} && (\mathbf{1},1) \ar@{~>}@/_/[ll]_r } \right)  
= \displaystyle{ c \sum_{x\in X} p_x \ln\left(\frac{p_x}{r_x}\right)  }
\]
This is the desired result, Equation~\eqref{showF}.
\end{proof}

\begin{lem}
\label{lem:4.2}
Equation~\eqref{showF} holds if $c < \infty$ and $\supp(p) \subseteq \supp (r)$.
\end{lem}

\begin{proof}
This can be reduced to the previous case by considering the commutative triangle
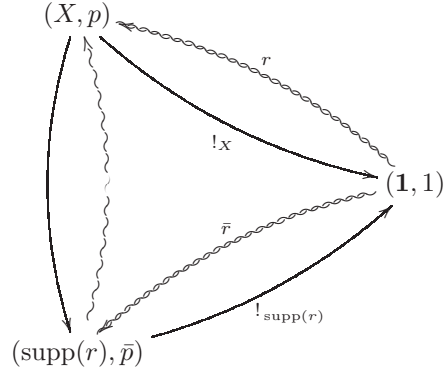
\[
\xymatrix@!=1.4cm{ (X,p) \ar@/_.4cm/[drr]_{!_{X}} \ar@/_.4cm/[dd] \\
 && (\mathbf{1},1) \ar@{~>}@/_.4cm/[ull]_{r} \ar@{~>}@/_.4cm/[dll]_{\bar{r}} \\
 (\supp(r),\bar{p}) \ar@{~>}@/_.4cm/[uu] \ar@/_.4cm/[urr]_{\;\; !_{\,\supp(r)}} }
\]
in which $\bar{p} =p|_{\supp(r)}$ and $\bar{r} =r|_{\supp(r)}$, and the vertical morphism consists of any map $X\to \supp(r)$ that restricts to the identity on  $\supp(r)$ and, as its stochastic right inverse, the inclusion $\supp(r)\hookrightarrow X$. This morphism lies in $\FP$.
\end{proof}

\begin{lem}
\label{lem:4.3}
Equation~\eqref{showF} holds if $0<c<\infty$.
\end{lem}
 
\begin{proof}
We already know by Lemma \ref{lem:4.2} that this holds
when $\supp (p) \subseteq \supp (r)$,
so assume otherwise.  Our task is then show that 
\[  F\left( \xymatrix{ (X,p) \ar@/_/[rr]_{!_X} && (\mathbf{1},1) \ar@{~>}@/_/[ll]_r } \right) = \infty  .
\]
To do this, choose $x \in X$ with $p_x>0=r_x$, and consider the commutative triangle
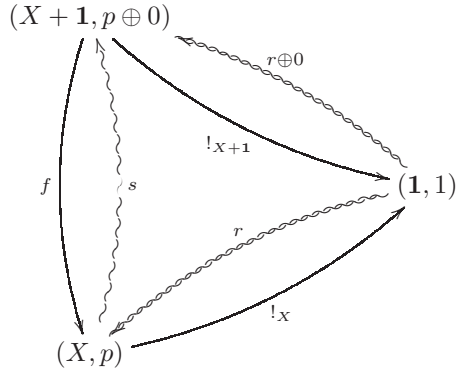
\[
\xymatrix@!=1.4cm{ (X + \mathbf{1},p\oplus 0) \ar@/_.4cm/[drr]_{!_{X + \mathbf{1}}} \ar@/_.4cm/[dd]_f \\
 && (\mathbf{1},1) \ar@{~>}@/_.4cm/[ull]_{r\oplus 0} \ar@{~>}@/_.4cm/[dll]_{r} \\
 (X,p) \ar@{~>}@/_.4cm/[uu]_s \ar@/_.4cm/[urr]_{!_X} }
\]
in which $f$ maps $X$ to itself by the identity and sends the unique element of $\mathbf{1}$ to $x$. This function has a one-parameter family of stochastic right inverses, and we take the arrow $s \maps X\leadsto X + \mathbf{1}$ to be any element of this family.  

To construct these stochastic right inverses, let $Y = X - \{x\}$.  This set is nonempty because the probability distribution $r$ is supported on it.  If $p_x < 1$ let $q$ be 
the probability distribution on $Y$ given by
\[         q = \frac{1}{1 - p_x} p|_Y  , \]
while if $p_x = 1$ let $q$ be an arbitrary probability distribution on $Y$.
For any $\alpha \in [0,1]$, the convex linear combination
\begin{equation}
\label{eq:convex_comb}
(1-p_x) \left(\!\! \xymatrix{ (Y, q) \ar@/_/[rr]_{1_Y} && (Y,  q) \ar@{~>}@/_/[ll]_{1_Y} } \!\!\right) \oplus 
p_x \left(\!\!\xymatrix{ (\mathbf{2},(1,0)) \ar@/_/[rr]_{!_\mathbf{2}} && (\mathbf{1},1) \ar@{~>}@/_/[ll]_{q(\alpha)} } \!\! \right)
\end{equation}
is a morphism in $\FinStat$.   There is a natural isomorphism from
its domain to that of the desired morphism $(f,s)$:
\[          (1 - p_x) (Y,q) \; \oplus \; p_x (\mathbf{2}, (1,0)) 
\; \cong \; (X + \mathbf{1}, p \oplus 0)    \]
and similarly for its codomain:
\[        (1 - p_x) (Y,q) \; \oplus \; p_x (\mathbf{1}, 1) 
\; \cong \; (X, p ) .   \]
Composing \eqref{eq:convex_comb} with these fore and aft, we obtain 
the desired morphism
\[   \xymatrix{ (X + \mathbf{1}, p \oplus 0) \ar@/_.2cm/[rr]_{s} && (X,  p) \ar@{~>}@/_.2cm/[ll]_{f} } . \]

Using convex linearity and the fact that $F$ vanishes on isomorphisms, 
\eqref{eq:convex_comb} implies that $F(f,s) = - p_x c \ln\alpha$.  Applying $F$ to our commutative triangle, we thus obtain 
\[
F\left( \xymatrix{ (X + \mathbf{1},p \oplus 0) \ar@/_/[rr]_{!_{X + \mathbf{1}}} && (\mathbf{1},1) \ar@{~>}@/_/[ll]_{r \oplus 0} } \right) = -p_x c \ln\alpha + F\left( \xymatrix{ (X,p) \ar@/_/[rr]_{!_X} && (\mathbf{1},1) \ar@{~>}@/_/[ll]_r } \right).
\]
Since $p_x, c>0$, the first term on the right-hand side depends on $\alpha$, but no other terms do. This is only possible if both other terms are infinite. This proves
\[
F\left( \xymatrix{ (X,p) \ar@/_/[rr]_{!_X} && (\mathbf{1},1) \ar@{~>}@/_/[ll]_r } \right) = \infty  ,
\]
as was to be shown.
\end{proof}

\begin{lem}
\label{lem:4.4}
Equation~\eqref{showF} holds if $c = 0$.
\end{lem}

\begin{proof}
That~\eqref{showF} holds in this case is a simple consequence of lower semicontinuity: approximate $r$ by a family of probability distributions whose support is all of $X$.  By Lemma \ref{lem:4.3}, $F$ maps all the resulting morphisms to $0$.  Thus, the same must be true for the original $r$.
\end{proof}

To conclude the proof of Lemma~\ref{lem:converse}, we need to show
Equation~\eqref{showF} holds if $c = \infty$.   To do this, it suffices to assume $c = \infty$ and show that
\[  F\left( \xymatrix{ (X,p) \ar@/_/[rr]_{!_X} && (\mathbf{1},1) \ar@{~>}@/_/[ll]_r } \right) = \infty   \]
whenever $p \ne r$.  The reasoning in the previous lemmas will not help 
us now, since in Lemma \ref{lem:4.1} we needed $c < \infty$.  As we shall see in Proposition~\ref{counterex}, the proof for $c = \infty$ must use lower semicontinuity.  However, since lower semicontinuity only produces an \emph{upper} bound on the value of $F$ at a limit point, it will have to be used in proving the contrapositive statement: if $F$ is finite on some morphism of the above form with $p \ne r$, then it is finite on some morphism of the form~\eqref{g}. Now in order to infer that the value of $F$ at the limit point of a converging family of distributions is finite, it is not enough to know that the value of $F$ is finite at each element of the family: one needs a \emph{uniform} bound.  The need to derive such a uniform bound is the reason for the complexity of the following argument.

In what follows we assume that $p$ and $r$ are probability distributions
on $X$ with $p \ne r$ and
\[  F\left( \xymatrix{ (X,p) \ar@/_/[rr]_{!_X} && (\mathbf{1},1) \ar@{~>}@/_/[ll]_r } \right) < \infty .  \]
We develop a series of consequences culminating in Lemma 
\ref{lem:claim6}, in which we see that $g(\alpha)$ is finite for 
some $\alpha <1$.  This implies $c < \infty$, thus demonstrating
the contrapositive of our claim that Equation~\eqref{showF} holds if 
$c = \infty$.

\begin{lem}
\label{lem:claim1}
There exist $\alpha,\beta\in[0,1]$ with $\alpha\neq\beta$ such that
\beq
\label{alphabeta}
h(\alpha,\beta) = F\left( \xymatrix{ (\mathbf{2},q(\alpha)) \ar@/_/[rr]_{!_{\mathbf{2}}} && (\mathbf{1},1) \ar@{~>}@/_/[ll]_{q(\beta)} } \right)
\eeq
is finite.
\end{lem}

\begin{proof} 
Choose some $y\in X$ with $p_y\neq r_y$, and define
$f \maps X \to \mathbf{2}$ by
\[
f(x) = \begin{cases} 1 & \textrm{ if } x = y \\ 0 & \textrm{ if } x\neq y. \end{cases}
\]
Put $\beta =1-r_y$. Then $f$ has a stochastic right inverse $s$ given by
\[
s_{x j} = \begin{cases} \displaystyle{\frac{r_x}{\beta}(1 - \delta_{x y})} & \textrm{ if } j=0 \\  \delta_{x y} & \textrm{ if } j=1 \end{cases}
\]
where, if $\beta=0$, we interpret the fractions as forming an arbitrarily chosen probability distribution on $X - \{y\}$. Setting $\alpha =1-p_y$, we have a commutative triangle
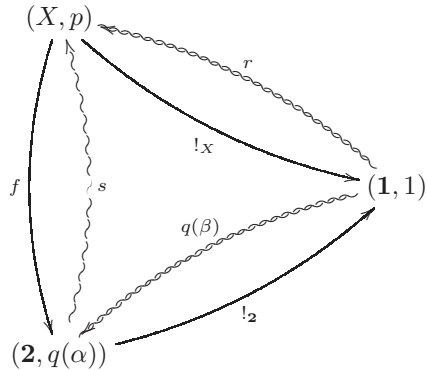
\[
\xymatrix@!=1.4cm{ (X,p) \ar@/_.4cm/[drr]_{!_{X}} \ar@/_.4cm/[dd]_f \\
 && (\mathbf{1},1) \ar@{~>}@/_.4cm/[ull]_{r} \ar@{~>}@/_.4cm/[dll]_{q(\beta)} \\
 (\mathbf{2},q(\alpha)) \ar@{~>}@/_.4cm/[uu]_s \ar@/_.4cm/[urr]_{!_{\mathbf{2}}} }
\]
and the claim follows from functoriality.
\end{proof}

\begin{lem}
\label{lem:claim2}
$h(\alpha',\tfrac{1}{2})$ is finite for some $\alpha'<\tfrac{1}{2}$.
\end{lem}

\begin{proof} Choose $\alpha, \beta$ as in Lemma \ref{lem:claim1}.  Consider the commutative square
\[
\xymatrix@!=1.5cm{ (\mathbf{4},\tfrac{1}{2}q(\alpha)\oplus\tfrac{1}{2} q(\beta)) \ar@/_.4cm/[rr]_{\scriptsize{\begin{array}{c}0,1\mapsto 0\\ 2,3\mapsto 1\end{array}}} \ar@/_.4cm/[dd]_{\scriptsize{\begin{array}{c}0,2\mapsto 0\\ 1,3\mapsto 1\end{array}}} && (\mathbf{2},q(\tfrac{1}{2})) \ar@/_.4cm/[dd]_{!_{\mathbf{2}}} \ar@{~>}@/_.4cm/[ll]_s \\\\
 (\mathbf{2},q(\tfrac{\alpha+\beta}{2})) \ar@/_.4cm/[rr]_{!_{\mathbf{2}}} \ar@{~>}@/_.4cm/[uu]_t && (\mathbf{1},1) \ar@{~>}@/_.4cm/[uu]_{q\left(\tfrac{1}{2}\right)} \ar@{~>}@/_.4cm/[ll]_{q(\beta)} }
\]
with the stochastic matrices
\[
s = \left( \begin{array}{cc} \beta & 0 \\ 1-\beta & 0 \\ 0 & \beta \\ 0 & 1-\beta \end{array} \right) = q(\beta) \oplus q(\beta) ,\qquad t = \left( \begin{array}{cc} \tfrac{1}{2} & 0 \\ 0 & \tfrac{1}{2} \\ \tfrac{1}{2} & 0 \\ 0 & \tfrac{1}{2} \end{array} \right).
\]

The right vertical morphism in this square lies in $\FP$, so $F$ vanishes on this.  The top horizontal morphism is a convex linear combination
\[
\frac{1}{2} \left( 
\xymatrix{ (\mathbf{2},q(\alpha)) \ar@/_/[rr]_{!_{\mathbf{2}}} && (\mathbf{1},1) \ar@{~>}@/_/[ll]_{q(\beta)} } \right)
\oplus 
\frac{1}{2} \left( 
\xymatrix{ (\mathbf{2},q(\beta)) \ar@/_/[rr]_{!_{\mathbf{2}}} && (\mathbf{1},1) \ar@{~>}@/_/[ll]_{q(\beta)} } \right) ,
\]
where the second term is in $\FP$.  Thus, by convex linearity and Lemma~\ref{lem:claim1}, $F$ of the top horizontal morphism equals $\frac{1}{2} h(\alpha,\beta) < \infty$.  By functoriality, $F$ is $\frac{1}{2} h(\alpha,\beta)$ on the composite of the top and right morphisms.  

This implies that the value of $F$ on the other two morphisms in the square must also be finite.  Let us compute $F$ of their composite in another way.  By definition, $F$ of the bottom horizontal morphism is $h(\frac{\alpha+\beta}{2},\beta)$.   The left vertical morphism is a convex linear combination
\[     \frac{\alpha + \beta}{2} \left(\!\!\!\xymatrix{ (\mathbf{2},q(\frac{\alpha}{\alpha+\beta})) \ar@/_/[rr]_{!_{\mathbf{2}}} && (\mathbf{1},1) \ar@{~>}@/_/[ll]_{q(\frac{1}{2})} } \!\!\! \right)  \oplus 
\frac{2 - \alpha - \beta}{2}  \left(\!\!\! \xymatrix{ (\mathbf{2},q(\frac{1-\alpha}{2-\alpha-\beta})) \ar@/_/[rr]_{!_{\mathbf{2}}} && (\mathbf{1},1) \ar@{~>}@/_/[ll]_{q(\frac{1}{2})} }\!\!\! \right) .
\]
By functoriality and convex linearity, $F$ on the composite of these two morphisms is thus
\[
\frac{\alpha+\beta}{2}\cdot h\!\left(\frac{\alpha}{\alpha+\beta},\frac{1}{2}\right) + \frac{2-\alpha-\beta}{2}\cdot h\!\left(\frac{1-\alpha}{2-\alpha-\beta},\frac{1}{2}\right) + h\left(\frac{\alpha+\beta}{2},\beta\right).
\]

Comparing these computations, we obtain
\begin{align}
\begin{split}
\label{feq}
h(\alpha,\beta) = ( & \alpha  +\beta)\cdot h\!\left(\frac{\alpha}{\alpha+\beta},\frac{1}{2}\right) \\
& + (2-\alpha-\beta)\cdot h\!\left(\frac{1-\alpha}{2-\alpha-\beta},\frac{1}{2}\right) + 2\cdot h\!\left(\frac{\alpha+\beta}{2},\beta\right).
\end{split}
\end{align}
This shows that each term on the right-hand side must be finite.
Note that the coefficients in front of these terms do not vanish, since $\alpha\neq\beta$.    If $\alpha < \beta$ then we can take $\alpha' =\tfrac{\alpha}{\alpha+\beta}$, so that
$\alpha' < \tfrac{1}{2}$, and the first term on the right-hand side gives $h(\alpha', \tfrac{1}{2}) < \infty$.    If $\alpha > \beta$ we can take $\alpha' = \tfrac{1 - \alpha}{2 - \alpha - \beta}$, so that $\alpha' < \tfrac{1}{2}$, and the second term on the right-hand side gives that $h(\alpha', \tfrac{1}{2}) < \infty$. 
\end{proof}

\begin{lem}
\label{lem:claim3}
For $\alpha\leq\beta\leq\tfrac{1}{2}$, we have $h(\beta,\tfrac{1}{2})\leq h(\alpha,\tfrac{1}{2})$.
\end{lem}

\begin{proof} By the intermediate value theorem, there exists $\gamma\in[0,1]$ with
\[
\gamma \alpha + (1-\gamma)(1-\alpha) = \beta.
\]
Now let $q(\alpha)\otimes q(\gamma)$ stand for the distribution on $\mathbf{4}$ with weights $(\alpha\gamma,\alpha(1-\gamma),(1-\alpha)\gamma,(1-\alpha)(1-\gamma))$.
The equation above guarantees that the left vertical morphism in this square is well-defined:
\[
\xymatrix@!=1.4cm{ (\mathbf{4},q(\alpha)\otimes q(\gamma)) \ar@/_.4cm/[rr]_{\scriptsize{\begin{array}{c}0,1\mapsto 0\\ 2,3\mapsto 1\end{array}}} \ar@/_.4cm/[dd]_{\scriptsize{\begin{array}{c}0,3\mapsto 0\\ 1,2\mapsto 1\end{array}}} && (\mathbf{2},q(\alpha)) \ar@/_.4cm/[dd]_{!_\mathbf{2}} \ar@{~>}@/_.4cm/[ll]_s \\\\
 (\mathbf{2},q(\beta)) \ar@{~>}@/_.4cm/[uu]_{t} \ar@/_.4cm/[rr]_{!_{\mathbf{2}}} && (\mathbf{1},1) \ar@{~>}@/_.4cm/[ll]_{q\left(\tfrac{1}{2}\right)} \ar@{~>}@/_.4cm/[uu]_{q\left(\tfrac{1}{2}\right)} }
\]
where we take:
\[
s = \left( \begin{array}{cc} \gamma & 0 \\ 1-\gamma & 0 \\ 0 & \gamma \\ 0 & 1-\gamma \end{array} \right), \qquad t = \left( \begin{array}{cc} \gamma & 0 \\ 0 & 1-\gamma \\ 0 & \gamma \\ 1-\gamma & 0 \end{array} \right)
\]
The square commutes and the upper horizontal morphism is in $\FP$, so the value of $F$ on the bottom horizontal morphism is bounded by the value of $F$ on the right vertical one, as was to be shown.  
\end{proof}

In the preceding lemma we are not yet claiming that $h(\alpha, \tfrac{1}{2})$ is finite.  We show this for $\alpha = \tfrac{1}{4}$ in Lemma~\ref{lem:claim4}, and for all $\alpha \in (0,1)$ in Lemma~\ref{lem:claim5}, where we actually obtain a uniform bound.

\begin{lem}
\label{lem:symmetry}
$h(\alpha,\tfrac{1}{2}) = h(1-\alpha,\tfrac{1}{2})$ for all $\alpha\in[0,1]$.
\end{lem}

\begin{proof}
Apply functoriality to the commutative triangle
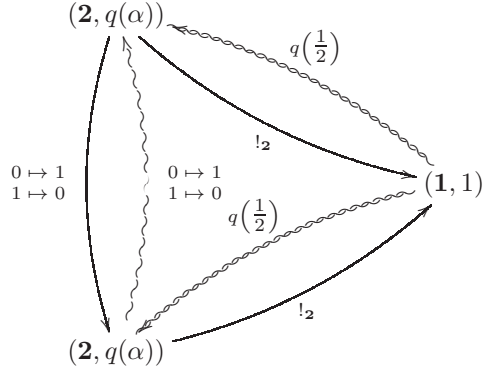
\[
\xymatrix@!=1.4cm{ (\mathbf{2},q(\alpha)) \ar@/_.4cm/[drr]_{!_{\mathbf{2}}} \ar@/_.4cm/[dd]_{\scriptsize{\begin{array}{c}0\mapsto 1\\ 1\mapsto 0\end{array}}} \\
 && (\mathbf{1},1) \ar@{~>}@/_.4cm/[ull]_{q\left(\tfrac{1}{2}\right)} \ar@{~>}@/_.4cm/[dll]_{q\left(\tfrac{1}{2}\right)} \\
 (\mathbf{2},q(\alpha)) \ar@{~>}@/_.4cm/[uu]_{\scriptsize{\begin{array}{c}0\mapsto 1\\ 1\mapsto 0\end{array}}} \ar@/_.4cm/[urr]_{!_{\mathbf{2}}} }
\]
where the vertical morphism is in $\FP$.
\end{proof}

\begin{lem}
\label{lem:claim4} $h(\tfrac{1}{4},\tfrac{1}{2})<\infty$.
\end{lem}

\begin{proof} We use~\eqref{feq} with $\beta=\tfrac{1}{2}$:

\begin{align}
\begin{split}
\label{feq2}
h\!\left(\alpha,\frac{1}{2}\right) = & \left( \alpha + \frac{1}{2}\right) h\!\left(\frac{2\alpha}{1+2\alpha},\frac{1}{2}\right) \\
& + \left(\frac{3}{2}-\alpha\right) h\!\left(\frac{2-2\alpha}{3-2\alpha},\frac{1}{2}\right) + 2 h\!\left(\frac{1+2\alpha}{4},\frac{1}{2}\right) ,
\end{split}
\end{align}
which we will apply for $\alpha<\tfrac{1}{2}$. On the right-hand side here, the first argument of $h$ in the second term can be replaced by $\tfrac{1}{3-2\alpha}$, thanks to Lemma~\ref{lem:symmetry}. Then the first arguments in all three terms on the right-hand side are in $[0,\tfrac{1}{2}]$, with the smallest in the first term, so Lemma~\ref{lem:claim3} tells us that
\[
h\!\left(\alpha,\frac{1}{2}\right) \leq 4  h\!\left(\frac{2\alpha}{1+2\alpha},\frac{1}{2}\right) .
\]
Now with $\alpha_0 =\tfrac{1}{4}$, the sequence recursively defined by $\alpha_{n+1} =\frac{2\alpha_n}{1+2\alpha_n}$ increases and converges to $\tfrac{1}{2}$. In particular we can find $n$ with $\alpha'<\alpha_n<\tfrac{1}{2}$, where $\alpha'$ is chosen as in Lemma~\ref{lem:claim2}.  Using that result together with Lemma~\ref{lem:claim3}, we obtain
\[
h\!\left(\frac{1}{4},\frac{1}{2}\right) \leq 4^n \, h\!\left(\alpha_n,\frac{1}{2}\right) \leq 4^n \, h\!\left(\alpha',\frac{1}{2}\right) < \infty.  \qedhere
\]
\end{proof}

\begin{lem}
\label{lem:claim5}
There is a constant $B < \infty$ such that 
$h(\alpha,\tfrac{1}{2})\leq B\, h(\tfrac{1}{4},\tfrac{1}{2})$ for all $\alpha\in(0,1)$.
\end{lem}

\begin{proof} 
By the symmetry in Lemma~\ref{lem:symmetry}, it is sufficient to consider $\alpha\in(0,\tfrac{1}{2}]$. By Lemma~\ref{lem:claim3}, we may use the bound $B=1$ for all $\alpha\in[\tfrac{1}{4},\tfrac{1}{2}]$. It thus remains to find a choice of $B$ that works  for all $\alpha\in(0,\tfrac{1}{4})$, and we assume $\alpha$ to lie in this interval from now on.

We reuse Equation~\eqref{feq2}. Both the second and the third term on the right-hand side have their first argument of $h$ in the interval $[\tfrac{1}{4},\tfrac{3}{4}]$, so we can apply Lemmas~\ref{lem:claim3} and \ref{lem:symmetry} to obtain
\[
h\!\left(\alpha,\frac{1}{2}\right) \leq \left(\alpha + \frac{1}{2}\right) 
h\!\left(\frac{2\alpha}{1+2\alpha},\frac{1}{2}\right) + 
\left(\frac{7}{2} - \alpha\right) h\!\left(\frac{1}{4},\frac{1}{2}\right) .
\]
To find a simpler-looking upper bound, we bound the right-hand side from above by applying Lemma~\ref{lem:claim3} in order to replace the $\tfrac{2\alpha}{1+2\alpha}$ argument by just $2\alpha$, and at the same time use $\alpha\in(0,\tfrac{1}{4})$ in order to bound the coefficients of both terms by $\alpha + \tfrac{1}{2}\leq \tfrac{3}{4}$ and $\tfrac{7}{2} - \alpha\leq \tfrac{7}{2}$:
\[
h\!\left(\alpha,\frac{1}{2}\right) \leq \frac{3}{4}\, h\!\left( 2\alpha, \frac{1}{2}\right) + \frac{7}{2} \, h\!\left(\frac{1}{4},\frac{1}{2}\right).
\]
If we put $\alpha = 2^{-n}$ for $n\geq 2$, then we can apply this inequality repeatedly until only terms of the form $h(\tfrac{1}{4},\tfrac{1}{2})$ are left. This results in a geometric series:
\[
h\!\left(2^{-n},\frac{1}{2}\right) \leq \left( \left(\frac{3}{4}\right)^{n-2} + \sum_{k=0}^{n-3} \left(\frac{3}{4}\right)^k \cdot\frac{7}{2}  \right)  
h\!\left(\frac{1}{4},\frac{1}{2}\right).
\]
whose convergence (as $n\tends\infty$) implies the existence of a constant $B < \infty$
with 
\[    h(2^{-n},\tfrac{1}{2}) \leq B \,h(\tfrac{1}{4}, \tfrac{1}{2}) \]
for all $n \ge 2$. The present lemma then follows with the help of Lemma~\ref{lem:claim3}.
\end{proof}

\begin{lem}
\label{lem:claim6}
Equation~\eqref{showF} holds if $c = \infty$.
\end{lem}

\begin{proof}
By Lemma \ref{lem:claim5} and the lower semicontinuity of $h$, we see
that 
\[  g(\tfrac{1}{2}) = h(0,\tfrac{1}{2})<\infty \]
This  implies that the constant $c$ with $g(\alpha) = -c \ln \alpha$ has $c < \infty$.
Recall that we have shown this under the assumption that there exist probability
distributions $p$ and $r$ on a finite set $X$ with $p \ne r$ and
\[  F\left( \xymatrix{ (X,p) \ar@/_/[rr]_{!_X} && (\mathbf{1},1) \ar@{~>}@/_/[ll]_r } \right) < \infty .  \]
So, taking the contrapositive, we see that if $c = \infty$, then 
\[  F\left( \xymatrix{ (X,p) \ar@/_/[rr]_{!_X} && (\mathbf{1},1) \ar@{~>}@/_/[ll]_r } \right) = \infty   \]
whenever $p$ and $r$ are distinct probability distributions on $X$.  This proves
Equation~\eqref{showF} except in the case where $p = r$.  But in that case, both sides
vanish, since on the left we are taking $F$ of a morphism in $\FP$, and on the right
we obtain $\infty \cdot 0 = 0$.
\end{proof}

\section{Counterexamples and subtleties}

One might be tempted to think that our Theorem~\ref{thm1} also holds if one relaxes the lower semicontinuity assumption to measurability, upon equipping the hom-spaces of both $\FinStat$ and $[0,\infty]$ with their $\sigma$-algebras of Borel sets. For $[0,\infty]$, this $\sigma$-algebra is the usual Borel $\sigma$-algebra: the sets of the form $(a,\infty)$ are open and hence measurable, the sets of the form $[0,b]$ are closed and hence measurable, and therefore all half-open intervals $(a,b]$ are measurable, and these generate the standard Borel $\sigma$-algebra.  However, for Theorem~\ref{thm1}, mere measurability of the functor $F$ is not enough:

\begin{prop}
\label{counterex}
There is a functor $\FinStat\to [0,\infty]$ that is convex linear, measurable on hom-spaces, and vanishes on $\FP$, but is not a scalar multiple of relative entropy.
\end{prop}

\begin{proof}
We claim that one such functor $G\maps\FinStat\to [0,\infty]$ is given by
\[
G\bigg( \xymatrix{ (X,p) \ar@/_/[rr]_f && (Y,q) \ar@{~>}@/_/[ll]_s } \bigg) = \begin{cases} 0 & \textrm{if } \supp(p)=\supp(s\circ q) , \\ \infty & \textrm{if } \supp(p)\neq \supp(s\circ q). \end{cases} 
\]
This $G$ clearly vanishes on $\FP$. Since taking the support of a probability distribution is a lower semicontinuous and hence measurable function, the set of all morphisms obeying $\supp(p)=\supp(s\circ q)$ is also measurable, and hence $G$ is measurable.

Concerning functoriality, for a composable pair of morphisms
\[
\xymatrix{ (X,p) \ar@/_/[rr]_f && \ar@{~>}@/_/[ll]_s (Y,q) \ar@/_/[rr]_g && (Z,r), \ar@{~>}@/_/[ll]_t }
\]
we have
\[
\supp(p) = \supp(s \circ q),\quad \supp(q) = \supp(t\circ r) \quad \Longleftrightarrow \quad \supp(p) = \supp(s\circ t\circ r) .
\]
This proves functoriality. A similar argument proves convex linearity.
\end{proof}

As a measure of information gain, this functor $G$ is not hard to understand intuitively: we gain no information whenever the set of \emph{possible} outcomes is precisely the set that we expected; otherwise, we gain an infinite amount information.

Since the collection of all functors satisfying our hypotheses is closed under sums and scalar multiples and also contains the relative entropy functor, we actually obtain a whole family of such functors. For example, another one of these functors is $G'\maps\FinStat\to[0,\infty]$ given by
\[
G'\bigg( \xymatrix{ (X,p) \ar@/_/[rr]_f && (Y,q) \ar@{~>}@/_/[ll]_s } \bigg) = \begin{cases} S(p,s\circ q) & \textrm{if } \supp(p)=\supp(s\circ q) , \\ \infty & \textrm{if } \supp(p)\neq \supp(s\circ q). \end{cases} 
\]

Our original idea was to use the work of Petz~\cite{Petz,Petzbook1} to prove Theorem~\ref{thm1}. However, as it turned out, there is a gap in Petz's argument. Although his purported characterization concerns the quantum version of relative entropy, the first part of his proof in~\cite{Petz} treats the classical case.  If his proof were correct, it would prove this:

\begin{unthm} 
\label{thm:petz}
The relative entropy $S(p,r)$ for pairs of probability measures on the same finite set
such that $r$ has full support is characterized up to a multiplicative constant by these
properties:
\begin{enumerate} 
\item\label{condex} \define{Conditional expectation law}. Suppose $f \maps X \to Y$ is a function and $s \maps Y \leadsto X$ a stochastic map with $f \circ s= 1_Y$.  Given probability distributions $p$ and $r$ on $X$, and assuming that $r$ has full support and $r =s\circ f \circ r$, we have
\beq
\label{Petzglomc}
S(p,r) = S(f \circ p ,f \circ r) + S(p,s\circ f\circ p) .
\eeq
\item\label{perminv} \define{Invariance.} Given any bijection $f \maps X\to Y$ and probability distributions $p$, $r$ on $X$ such that $r$ has \define{full support} (i.e.\ its support is all of $X$), we have
\[
S(f \circ p, f \circ r) = S(p,r) .
\]
\item\label{convlin} \define{Convex linearity.}  Given probability distributions 
$p,r$ on $X$ and $p',r'$ on $Y$ such that $r$ and $r'$ have full support, 
and given $\lambda \in [0,1]$, we have
\[
S(\lambda p \oplus (1-\lambda) p',\lambda r \oplus (1-\lambda) r') \; = \; \lambda S(p,r) + (1-\lambda) S(p',r') .
\]
\item\label{nilpot} \define{Nilpotence.} For any probability distribution
$p$ with full support on a finite set, $S(p,p) = 0$. 
\item\label{meas} \define{Measurability property.} The function
\[
(p,r) \mapsto S(p,r)
\]
is measurable on the space of pairs of probability distributions on $X$ such that $r$ 
has full support.
\end{enumerate}
\end{unthm}
\noindent
Note that~\cite{Petz} uses the opposite ordering for the two arguments of $S$.

The problem with this ``theorem'' is the range of applicability of Equation~\eqref{Petzglomc}: what is this formula supposed to mean when $s\circ f\circ p$ does not have full support? After all, $S(p,r)$ is assumed to be defined only when the second argument has full support, but this need not be the case for $s\circ f\circ p$, given the assumptions made in the statement of the conditional expectation property. (Note that $f\circ r$ has full support, so the term $S(f\circ p,f\circ r)$ is fine.)

One can try to correct this problem by assuming that the conditional expectation property holds only if $s\circ f\circ p$ has full support as well. However, this means that the proof of Petz's Lemma 1 is valid only when (using his notation) $p_3 > 0$, which implies that his Equation (5) is known to hold only for $p_2>0$ and $p_3>0$. Upon following the thread of Petz's argument, one finds that his Equation (6) has been proven to follow from his assumptions only for $x\in(0,1)$ and $u\in (0,1)$. However, the solution of that functional equation in the references he points to crucially uses the assumption that the functional equation also holds in case that $x=0$ or $u=0$. This is the gap in Petz's proof.

In fact, if one allows $S$ to take on infinite values, then the above classical version of Petz's theorem is not even correct, if one uses the interpretation that~\eqref{Petzglomc} is to be applied only when $s\circ f\circ p$ has full support. The counterexample is similar to our functor $G'$ from above:
\[
S'(p,r) = \begin{cases} S(p,r) & \textrm{if $p$ has full support} , \\ \infty & \textrm{otherwise}. \end{cases} 
\]

\section{Conclusions}
The theorem here, and our earlier characterization of entropy \cite{BFL}, can be
seen as part of a program of demonstrating that mathematical  structures that are ``socially important'' are also ``categorically natural''.   Tom Leinster, whose words we
quote here, has carried this forward to a categorical explanation of Lebesgue
integration \cite{Leinster3}.  It would be interesting to generalize our results on 
entropy and relative entropy from finite sets to general measure spaces, where 
integrals replace sums.  It would be even more interesting to do this using 
a category-theoretic approach to integration.

It would also be good to express our theorem more concisely.   As noted in 
Appendix \ref{app:convalgs}, convex
linear combinations are operations in a topological operad $\P$.  We can define
`convex algebras', that is, algebras of $\P$, in any symmetric monoidal topological category. The category $[0,\infty]$ with the upper topology on its set of morphisms is a convex algebra in $\Top\Cat$, the (large) topological category of topological categories.  We believe, but have not proved, that $\FinStat$ is a `weak' convex algbra in $\Top\Cat$.  This would mean that the axioms for a convex algebra hold up to coherent natural isomorphism \cite{Leinster2}.   If this is true, the relative entropy
 \[   \RE \maps \FinStat \to [0,\infty]  \]
should be, up to a constant factor, the unique map of weak convex algebras that 
vanishes on morphisms in $\FP$.  Leinster \cite{Leinster} has shown that $\FP$ 
is also a weak convex algebra in $\Cat(\Top)$.  In fact, it is the free such thing
on an internal convex algebra.   So, it seems that both entropy and relative entropy
emerge naturally from a category-theoretic examination of convex linearity.  

\appendix

\section{Semicontinuous functors}
\label{app:semicont}

In Section \ref{sec:semicont} we explained what it meant for relative entropy to be
a semicontinuous functor.   A more sophisticated way to think about semicontinuous functors uses topological categories.   This requires that we put a nonstandard topology on $[0,\infty]$, the so-called `upper topology'.

A topological category is a category internal to $\Top$, and a continuous functor is a functor internal to $\Top$.   In other words:

\begin{defn}  A \define{topological category} $C$ is a small category where the set
of objects $C_0$ and the set of morphisms $C_1$ are equipped with the structure
of topological spaces, and the maps assigning to each morphism its source and target:
\[  s, t \maps C_1 \to C_0  \]
the map assigning to each object its identity morphism
\[   i \maps C_0 \to C_1 \]
and the map sending each pair of composable morphisms to their composite
\[   \circ \maps C_1 \times_{C_0} C_1 \to C_1 \]
are continuous.   Given topological categories $C$ and $D$, a \define{continuous
functor} is a functor $F \maps C \to D$ such that the map on objects 
$F_0 \maps C_0 \to D_0$ and the map on morphisms $F_1 \maps C_1 \to D_1$ are continuous.
\end{defn}

We now explain how $\FinStoch$ and $\FinStat$ are topological categories. Strictly speaking, in order for this to work, we need to deal with size issues.  One approach is to let the objects of $\Top$ be `large' sets living in a higher Grothendieck universe, which allows us to talk about the set of all objects or morphisms of $\FinStat$ or $\FinStoch$.  Another is to replace each of these categories by its skeleton, which is an equivalent small category.  From now on, we assume that one of these things has been done.

For $\FinStoch$, we put the discrete topology on its set of objects $\FinStoch_0$. Each hom-set $\FinStoch(X,Y)$ is a subset of the Euclidean space $\R^{|X|\times|Y|}$, and we put the subspace topology on this hom-set; for example, $\FinStoch(1,Y)$, the set of all probability distributions on $Y$, is topologized as a simplex. In this way, $\FinStoch$ becomes a category \textit{enriched} over $\Top$, and in particular internal to $\Top$. 

As for $\FinStat$, the identification
\[
\FinStat_0 = \left\{ (X,p) \:|\: X\in\FinStoch_0,\: p\in\FinStoch(1,X) \right\} \subseteq \FinStoch_0\times\FinStoch_1
\]
induces a topology on $\FinStat_0$. In this topology, a net $(X^\lambda,p^\lambda)_{\lambda \in\Lambda}$ converges to $(X,p)$ if and only if eventually $X^\lambda=X$, and $p^\lambda \to p$ for those $\lambda$ with $X^\lambda=X$. Similarly, every morphism in $\FinStat$ consists of a pair of morphisms in $\FinStoch$ satisfying certain conditions, and the resulting inclusion 
\[
\FinStat_1 \subseteq \FinStoch_1\times\FinStoch_1
\]
can be used to define a topology on $\FinStat_1$. We omit the verification that these topologies make $\FinStat$ into a topological category.

There is a topology on $[0,\infty]$ where the open sets are those of the form $(a,\infty]$, together with the whole space and the empty set.    This is called the \define{upper topology}.  With this topology, a function $\psi \maps A \to [0,\infty]$ from any topological space $A$ is continuous if and only $\psi$ is lower semicontinuous, meaning 
\[         \psi(a) \le \liminf_{\lambda\to \infty} \psi(a^\lambda)  \] 
for every convergent net $a^\lambda \in A$.  It is easy to check that this topology on $[0,\infty]$ makes addition continuous. 

In short, $[0,\infty]$ with its upper topology is a topological monoid under
addition.   We thus obtain a topological category with 
one object and $[0,\infty]$ as its topological monoid of endomorphisms.  By abuse of 
notation we also call this topological category simply $[0,\infty]$.  This lets us state Lemma \ref{lem:semicont} in a different way:

\begin{lem}
\label{lem:semicont2}
If $[0,\infty]$ is viewed as a topological category using the upper topology,  
the functor $\RE \maps \FinStat \to [0,\infty]$ is continuous.  
\end{lem}

On the other hand, if we give the monoid $[0,\infty]$ the less exotic topology where it is homeomorphic to a closed interval, then this functor is \emph{not} continuous.  

Having gone this far, we cannot resist pointing out that $[0,\infty]$ with its 
upper topology is also a topological rig.  Recall that a \define{rig} is a `ring without negatives': a set equipped with an addition making it into a commutative monoid and a multiplication making it into a monoid, with multiplication distributing over 
addition.  In other words, it is a monoid in the monoidal category of commutative monoids.  A \define{topological rig} is a rig with a topology in which addition and multiplication are continuous.  To make $[0,\infty]$ into a rig, we define addition as before, define multiplication in the usual way for numbers in $[0,\infty)$, and set
\[            0 a = a 0 = 0  \]
for all $a \in [0,\infty]$.  One can verify that multiplication is continuous: but
again, the key point is that we need to use the upper topology, since 
$\infty \cdot a$ suddenly jumps from $\infty$ to $0$ as $a$ reaches zero.  Thus:

\begin{lem}
\label{lem:rig}
With its upper topology, $[0,\infty]$ is a topological rig.
\end{lem}

More important now is that $[0,\infty]$ is a module over the rig 
$[0,\infty)$, where addition and multiplication in the latter are defined as usual
and we define the action of $[0,\infty)$ on $[0,\infty]$ using multiplication,
with the proviso that $0 \cdot a = 0$ even when $a = \infty$.  And here we
see:

\begin{lem}
\label{lem:topmodule}
The topological monoid $[0,\infty]$ with its upper topology becomes a 
topological module over the rig $[0,\infty)$ with its usual topology.
\end{lem}

\section{Convex algebras}
\label{app:convalgs}

We define the \define{monad for convex sets} to be the monad on $\Set$ sending any set $X$ to the set of finitely-supported probability distributions on $X$.  For example, this monad sends $\{1, \ldots, n\}$ to the set 
\[  \P_n = \{ p \in [0,1]^n : \; \sum_{i = 1}^n p_i = 1 \} \]
which can be identified with the $(n - 1)$-simplex.  This monad is finitary, so can be thought about in a few different ways.

First, a finitary monad can thought of as a finitary algebraic theory.  The monad for convex sets can be presented by a family $(*_\lambda)_{\lambda \in [0, 1]}$ of binary operations, subject to the equations
\begin{align*}
x *_0 y &= x, \\
x *_\lambda x &= x, \\
x *_\lambda y &= y *_{1-\lambda} x, \\
(x *_\mu y) *_\lambda z &= x *_{\lambda\mu} (y *_{\tfrac{\lambda(1-\mu)}{1-\lambda\mu}} z)
\end{align*}
For $\lambda=\mu=1$, the fraction $\tfrac{\lambda(1-\mu)}{1-\lambda\mu}$ in the last equation may be taken to be an arbitrary number in $[0,1]$. See~\cite{Fritz} for more detail on how to derive this presentation from the monad.

A finitary algebraic theory can also be thought of as an operad with extra structure.  In a symmetric operad $\O$, one has for each bijection $\sigma: \{1, \ldots, n\} \to \{1, \ldots, n\}$ an induced map $\sigma_* \maps \O_n \to \O_n$.  In a finitary algebraic theory, one has the same thing for \emph{arbitrary} functions between finite sets, not just bijections.  In other words, a finitary algebraic theory amounts to a non-symmetric operad $\O$ together with, for each function $\theta \maps \{1, \ldots, m\} \to \{1, \ldots, n\}$ between finite sets, an induced map $\theta_* \maps \O_m \to \O_n$, satisfying suitable axioms. 

\begin{defn}
The underlying symmetric operad for the monad for convex sets is called the \define{operad for convex algebras} and denoted $\P$.  An algebra of $\P$ is called a \define{convex algebra}.
\end{defn}

The space of $n$-ary operations for this operad is $\P_n$, the space of probability distributions on $\{1, \ldots, n\}$.    The composition of operations works as follows.   Given probability distributions $p \in \P_n$ and $r_i \in \P_{k_i}$  for each $i \in \{1, \ldots, n\}$, we obtain a probability distribution 
$p \circ (r_1, \dots, r_n) \in \P_{k_1+ \cdots + k_n}$, namely
\[      p \circ (r_1, \dots, r_n) = (p_1r_{11}\dots,p_1 r_{1k_1}, \dots 
p_n  r_{n1}, \dots , p_n r_{nk_n}) . \]
The maps $\theta_* \maps \P_{m} \to \P_{n}$ can be defined by pushforward of measures.  An algebra for the algebraic theory of convex algebras is an algebra $X$ for the operad with the further property that the square
\[
\xymatrix{
\P_m \times X^n \ar[rr]^{1 \times \theta^*} \ar[dd]_{\theta_*\times 1 }
&&   \P_m \times X^m \ar[dd]
\\  \\
\P_n \times X^n   \ar[rr]^{}                    && X
}
\]
commutes for all $\theta \maps \{1, \ldots, m\} \to \{1, \ldots, n\}$, where the unlabelled arrows are given by the convex algebra structure of $X$.

Note that $\P$ is naturally a topological operad, where the topology on $\P_n$ is the usual topology on the $(n-1)$-simplex.   In this paper we have implicitly been using
algebras of $\P$ in various topological categories $\E$ 
with finite products.  We call these \define{convex algebras} in $\E$.  Here are some examples:

\begin{itemize}
\item Any convex subset of $\R^n$ is a convex algebra in $\Top$.  
\item The additive monoid $[0,\infty]$ with its upper topology becomes a convex algebra in $\Top$ if we define convex linear combinations by treating $[0,\infty]$ as a topological module of the rig $[0,\infty)$ as in Lemma~\ref{lem:topmodule}.  We must equip $[0,\infty]$ with its upper topology for this to work, because the convex linear combination $\lambda \cdot \infty + (1 - \lambda) \cdot a$ equals $\infty$ when $\lambda > 0$, but suddenly jumps down to $a$ when $\lambda$ reaches zero.

\item The category $\Cat(\Top)$ of small topological categories and continuous functors is itself a large topological category.  If we regard $[0,\infty]$ with its upper topology as a one-object topological category as in Appendix~\ref{app:semicont}, then it becomes a convex algebra in $\Cat(\Top)$ thanks to the previous remark. 

\item The categories $\FinProb$, $\FinStat$ should be `weak convex algebras'
in $\Cat(\Top)$, though we have not carefully checked this.  By this, we mean that axioms for an algebra of the operad $\P$ hold up to coherent natural isomorphism, in the sense made precise by Leinster \cite{Leinster2}. 

\item Similarly, Leinster has shown that $\FP$ is a weak convex algebra in 
$\Cat(\Top)$.   In fact, it is equivalent to the free convex algebra in $\Cat(\Top)$ 
on an internal convex algebra \cite{Leinster}. 

\end{itemize}

\subsection*{{\bf Acknowledgements.}} We thank Ryszard Kostecki and Rob Spekkens for discussions and an unintended benefit. TF was supported by Perimeter Institute for Theoretical Physics through a grant from the John Templeton foundation. Research at Perimeter Institute is supported by the Government of Canada through Industry Canada and by the Province of Ontario through the Ministry of Research and Innovation. JB thanks the Centre for Quantum Technologies for their support.

\end{document}